\documentclass{article}
\usepackage{graphicx} % Required for inserting images
\usepackage{amssymb}
\usepackage{amscd}
\usepackage{caption}
\usepackage{blkarray}
\usepackage{mathrsfs}
\usepackage{xcolor}
\usepackage{bbm}
\usepackage[centertags]{amsmath}
\usepackage{amssymb}
\usepackage[english]{babel}
\usepackage{blindtext}
\usepackage{tikz}
\usepackage{amsthm}
\usepackage{graphicx}
\usepackage{gnuplottex}
\usepackage{epsfig}
\usepackage[belowskip=-15pt,aboveskip=0pt]{caption}
\usepackage{amsthm}
\usepackage{amssymb,latexsym}
\usepackage{amsfonts,amsmath}
\usepackage[T1]{fontenc}
\usepackage{amsmath,amssymb}
\usepackage{latexsym}
\numberwithin{equation}{section}
\usepackage{amssymb}
\usepackage{amsfonts}
\usepackage[arrow,frame,matrix]{xy}

\usepackage{fullpage}
\usepackage{hyperref}
\usepackage[margin=2.6cm]{geometry}
\usepackage{placeins}
\newtheorem{definition}{Definition}

\usepackage{grffile}
\usetikzlibrary{decorations.pathreplacing}
\usepackage{color}
\usepackage{float}
\usepackage{subcaption}
\usepackage{cancel}
\usepackage{authblk}
\usepackage{ragged2e}
\usepackage{bm}
\usepackage[strings]{underscore}
\usepackage{parskip}
\newtheorem{theorem}{Theorem}
\newtheorem{remark}{Remark}

\begin{document}
\title{\bf  GLT hidden structures in mean-field quantum spin systems}
\author[1]{Christiaan J. F. van de Ven}
\author[2]{Muhammad Faisal Khan}
\author[2,3]{S. Serra-Capizzano}
\affil[1]{Friedrich-Alexander-Universit\"{a}t Erlangen-N\"{u}rnberg, Department of Mathematics\\ Cauerstra\ss e 11, 91058 Erlangen (Germany)}
\affil[2]{University of Insubria, Department of Science and High Technology\\ via Valleggio 11, 22100  Como (Italy)}
\affil[3]{University of Uppsala, Department of Information Technology\\ hus 10, L\"{a}gerhyddsv\"{a}gen 1, 75105 Uppsala (Sweden)}
\date{}

\maketitle

\begin{abstract}
\noindent
This work explores structured matrix sequences arising in mean-field quantum spin systems. We express these sequences within the framework of generalized locally Toeplitz (GLT) $*$-algebras, leveraging the fact that each GLT matrix sequence has a unique GLT symbol. This symbol characterizes both the asymptotic singular value distribution and, for Hermitian or quasi-Hermitian sequences, the asymptotic spectral distribution. Specifically, we analyze two cases of real symmetric matrix sequences stemming from mean-field quantum spin systems and determine their associated distributions using GLT theory. Our study concludes with visualizations and numerical tests that validate the theoretical findings, followed by a discussion of open problems and future directions.
%The present work deals with structured matrix-sequences, arising when considering mean-field quantum spin systems. The idea is to express the related matrix-sequences in the formalism of the theory of generalized locally Toeplitz (GLT) $*$-algebras. The gain is represented by the fact that each GLT matrix-sequence has a unique GLT symbol and the GLT symbol represents the asymptotic singular value distribution of the considered GLT matrix-sequence and also its asymptotic spectral distribution, when the matrix-sequence is Hermitian or quasi-Hermitian. In the current study we analyze two cases of real symmetric matrix-sequences stemming from mean-field quantum spin systems, and we find the associated distributions via the GLT theory. We end the work with visualizations and numerical tests corroborating the theoretical findings, followed by concluding remarks and open problems.


\smallskip

\end{abstract}

\section{Introduction}\label{sec: intro}
In recent years, the field of generalized locally Toeplitz (GLT) sequences  has experienced significant advancements, driven by their growing applicability in numerical analysis and applied sciences via highly performing computational algorithms; see \cite{GLT-1,GLT-2,Tilli-LT} for the seminal papers and  see \cite{barbarino2020uni,barbarino2020multi,BaSe,garoni2017,garoni2017} for the systematic treatment of the theory. These developments have not only broadened the theoretical foundation of GLT sequences, but also underscored their relevance in diverse scientific domains; see \cite{Appl-1-GLT,Appl-2-GLT,tom,Appl-3-GLT} for a selection of applications in biomedicine, economics, engineering.

The current paper explores the implications of GLT theory within the realm of mean-field quantum spin systems,  with a primary focus on the quantum Curie-Weiss model. As a paradigmatic example of mean-field interactions in quantum statistical mechanics, the Curie-Weiss model captures key emergent behaviors - such as spontaneous symmetry breaking and phase transitions - in the asymptotic regime of large system sizes \cite{Ven_2020,Ven_2022}. The motivation for this study stems from the mathematical complexity inherent in analyzing quantum spin models, especially in the thermodynamic limit where the number of lattice sites becomes large. While prior research has addressed various aspects of these systems, ranging from spectral properties to the decay of correlations and critical phenomena, a systematic understanding grounded in GLT theory is still lacking. GLT sequences offer a powerful and compact framework to describe the asymptotic behavior of large matrix sequences, making them particularly well-suited to tackle these challenges. By leveraging the structure of GLT theory, we aim to deepen the analytical understanding of mean-field models, which play a central role in condensed matter physics. This necessitates a multidimensional approach that integrates theoretical insights with physical interpretations, ensuring a comprehensive perspective on the problem. By combining probability theory, GLT theory, and numerical methods, this paper aims to establish a novel framework for treating mean-field quantum spin systems as GLT sequences, thereby contributing to both the advancement  of mathematical theory and its applications in physics.

The present article is organized as follows. In Section \ref{sec: tools} we present the main tools from asymptotic linear algebra and especially on the GLT theory. In Section \ref{sec: CW-model} and Section \ref{sec: our-pb} we study the spectral distribution of two matrix-sequences arising in the Curie-Weiss model, namely the unrestricted one and the restricted version. The spectral distributions are checked numerically in Section \ref{sec: num}, together with other more specific features by discussing the links with the analytical properties and the related GLT symbols. Section \ref{sec: end} contains conclusions and an indication of future fruitful research directions, that could stem from the current work.

\section{Matrix-sequences with explicit and hidden structures}\label{sec: tools}

In this section, we first provide formal definitions of spectral and singular value distributions. Then we present two classes of matrices with explicit structures, i.e., diagonal sampling matrices and Toeplitz matrices. Subsequently, we consider asymptotic notions which make sense only for matrix-sequences: we consider in fact the zero-distributed matrix-sequences and the $*$-algebras of $d$-level, $r$-block generalized locally Toeplitz (GLT) matrix-sequences, whose construction is involved \cite{GLT-1,GLT-2,Tilli-LT} and needs topological tools \cite{garoni-topology} like those related to approximating class of sequences \cite{acs,gacs}. Instead we briefly present an equivalent axiomatic characterization taken from the books \cite{garoni2017,garoni2018} for $r=1$, $d\ge 1$ and from the works \cite{barbarino2020uni,barbarino2020multi} for the general setting with $d,r\ge 1$. For our results on the specific problems in Section \ref{sec: CW-model} and in Section \ref{sec: our-pb} the case $d=r=1$ is sufficient. However, as stressed in the conclusions, the natural setting would require either $d>1$ or $r>1$.

\begin{definition}\label{99}
    (Singular Value and Eigenvalue Distribution of a Matrix-Sequence). Let $\{A_n\}_n$ be a matrix-sequence, with $A_n$ of size $d_n$, and let $\psi : D \subset \mathbb{R}^t \to \mathbb{C}$ be a measurable function defined on a set $D$ with $0 < \mu_t(D) < \infty$.
\begin{itemize}
    \item We say that $\{A_n\}_n$ has an (asymptotic) singular value distribution described by $\psi$, and we write $\{A_n\}_n \sim_\sigma \psi$, if
    \[
    \lim_{n \to \infty} \frac{1}{d_n} \sum_{i=1}^{d_n} F(\sigma_i(A_n)) = \frac{1}{\mu_t(D)} \int_D  F(|\psi(\bm{x})|) \, d\bm{x}, \quad \forall F \in C_c(\mathbb{R}).
    \]

    \item We say that $\{A_n\}_n$ has an (asymptotic) spectral (or eigenvalue) distribution described by $\psi$, and we write $\{A_n\}_n \sim_\lambda \psi$, if
    \[
    \lim_{n \to \infty} \frac{1}{d_n} \sum_{i=1}^{d_n} F(\lambda_i(A_n)) = \frac{1}{\mu_t(D)} \int_D  F(\psi(\bm{x})) \, d\bm{x}, \quad \forall F \in C_c(\mathbb{C}).
    \]

    \item If $\psi$ describes both the singular value and eigenvalue distribution of $\{A_n\}_n$, we write $\{A_n\}_n \sim_{\sigma, \lambda} \psi$.
\end{itemize}
    When $\{A_n\}_n \sim_\lambda \psi$, the function $\psi$ is referred to as the \textit{eigenvalue (or spectral) symbol} of $\{A_n\}_n$.
\end{definition}

\subsection{Zero-distributed sequences}\label{ssec: zero-dist}
Zero-distributed sequences are defined as matrix-sequences $\{A_n\}_n$ such that $\{A_n\}_n \sim_\sigma 0$.  The following theorem, taken from \cite{garoni2017}, provides a useful characterization for detecting this type of sequence. We use the natural convention $1/\infty = 0$.
\begin{theorem}
 Let $\{A_n\}_n$ be a matrix-sequence, with $A_n$ of size $d_n$ and let $p\in [1,\infty]$, with $\|X\|_p$ being the Schatten $p$ norm of $X$, that is the $l^p$ norm of the vector its singular values. Let $\|\cdot\|=\|\cdot\|_\infty$ be the spectral norm. Then
\begin{itemize}
    \item $\{A_n\}_n \sim_\sigma 0$ if and only if $A_n = R_n + N_n$ with $\text{rank}(R_n)/d_n \to 0$ and $\|N_n\| \to 0$ as $n \to \infty$;
    \item $\{A_n\}_n \sim_\sigma 0$ if there exists $p \in [1, \infty]$ such that $\|A_n\|_p/d_n^{1/p} \to 0$ as $n \to \infty$.
\end{itemize}
\end{theorem}

\subsection{Unilevel scalar Toeplitz matrices}\label{ssec: toep}
Given ${n} \in \mathbb{N}$, a matrix of the form
\[
[A_{{i}-{j}}]_{{i},{j}=1}^{{n}} \in \mathbb{C}^{{n} \times {n}},
\]
with entries $A_{{k}} \in \mathbb{C}$, ${k} \in [-({n-1}), \dots, {n-1}]$, is called Toeplitz matrix.

Given a matrix-valued function $f : [-\pi, \pi] \to \mathbb{C}$ belonging to $L^1([-\pi, \pi])$, the ${n}$-th Toeplitz matrix associated with $f$ is defined as
\[
T_{{n}}(f) := [\hat{f}_{{i-j}}]_{{i,j=1}}^{{n}} \in \mathbb{C}^{{n} \times {n}},
\]
where
\[
\hat{f}_{{k}} = \frac{1}{2\pi} \int_{[-\pi,\pi]} f({\theta}) e^{-i({k},{\theta)}} d {\theta} \in \mathbb{C}, \quad {k} \in \mathbb{Z},
\]
are the Fourier coefficients of $f$, in which $i$ denotes the imaginary unit.
%Equivalently, $T_{\bm{n}}(f)$ can be expressed as
%\[
%T_{\bm{n}}(f) = \sum_{|j_1|<n_1} \dots \sum_{|j_d|<n_d} J_{n_1}^{(j_1)} \otimes \dots \otimes J_{n_d}^{(j_d)} \otimes \hat{f}_{(j_1, \dots, j_d)},
%\]
%where $\otimes$ denotes the Kronecker tensor product between matrices and $J_m^{(l)}$ is the matrix of order $m$ whose $(i,j)$ entry equals 1 if $i - j = %l$ and zero otherwise.

$\{T_{{n}}(f)\}_{{n} \in \mathbb{N}}$ is the family of (multilevel block) Toeplitz matrices associated with $f$, which is called the \textit{generating function}.

\subsection{Diagonal sampling matrices}\label{ssec: sampl}
Given a function $a : [0,1] \to \mathbb{C}$, we define the diagonal sampling matrix $D_{{n}}(a)$ as the diagonal matrix
\[
D_{{n}}(a) = \text{diag}_{{i}={1}, \dots, {n}} \, a\left( \frac{{i}}{{n}} \right) \in \mathbb{C}^{{n} \times {n}}.
\]

\subsection{GLT Axioms}\label{ssec: GLT axioms}

For fixed $d,r\ge 1$, a $d$-level, $r$-block GLT sequence is a special matrix-sequence equipped with a measurable function $\kappa : [0, 1]^d \times [-\pi, \pi]^d \to \mathbb{C}^{r\times r}$, called symbol. The symbol is essentially unique, in the sense that if $\kappa, \varsigma$ are two symbols of the same GLT sequence, then $\kappa = \varsigma$ a.e. We write $\{A_n\}_n \sim_{\mathrm{GLT}} \kappa$ to denote that $\{A_n\}_n$ is a GLT sequence with symbol $\kappa$.\\
For $d=r=1$, it can be proven that the set of GLT sequences is the $\ast$-algebra generated by the three classes of matrix-sequences defined in \ref{ssec: zero-dist}, \ref{ssec: toep}, \ref{ssec: sampl}, via the closure in the a.c.s. topology in \cite{garoni-topology}: zero-distributed, Toeplitz, and diagonal sampling matrix-sequences are indeed the generators of the induced $\ast$-algebra. The GLT classes satisfy several algebraic and topological properties that are treated in detail in \cite{barbarino2020uni,barbarino2020multi,garoni2017,garoni2018}. Here, we focus on the case of interest of $r=d=1$ and on the main operative properties, listed below, that represent a complete characterization of scalar unilevel GLT matrix-sequences \cite{garoni2017}, equivalent to the full constructive definition in \cite{GLT-1}.

\begin{itemize}
    \item \textbf{GLT 1.} If $\{A_n\}_n \sim_{\mathrm{GLT}} \kappa$ then $\{A_n\}_n \sim_\sigma \kappa$ in the sense of Definition \ref{99}, with $t = 2$ and $D = [0, 1] \times [-\pi, \pi]$. Moreover, if each $A_n$ is Hermitian, then $\{A_n\}_n \sim_\lambda \kappa$, again in the sense of Definition \ref{99} with $t = 2$.
    \item \textbf{GLT 2.} We have
    \begin{itemize}
        \item $\{T_{{n}}(f)\}_{{n}} \sim_{\mathrm{GLT}} \kappa({x}, {\theta}) = f({\theta)}$ if $f : [-\pi, \pi] \to \mathbb{C}$ is in $L^1([- \pi, \pi])$;
        \item $\{D_{{n}}(a)\}_{{n}} \sim_{\mathrm{GLT}} \kappa({x}, {\theta}) = a({x})$ if $a : [0, 1] \to \mathbb{C}$ is Riemann-integrable;
        \item $\{Z_n\}_n \sim_{\mathrm{GLT}} \kappa({x}, {\theta}) = 0$ if and only if $\{Z_n\}_n \sim_\sigma 0$.
    \end{itemize}

    \item \textbf{GLT 3.} If $\{A_n\}_n \sim_{\mathrm{GLT}} \kappa$ and $\{B_n\}_n \sim_{\mathrm{GLT}} \varsigma$, then
    \begin{itemize}
        \item $\{A_n^*\}_n \sim_{\mathrm{GLT}} \kappa^*$;
        \item $\{\alpha A_n + \beta B_n\}_n \sim_{\mathrm{GLT}} \alpha \kappa + \beta \varsigma$ for all $\alpha, \beta \in \mathbb{C}$;
        \item $\{A_n B_n\}_n \sim_{\mathrm{GLT}} \kappa \varsigma$;
        \item $\{A_n^\dagger\}_n \sim_{\mathrm{GLT}} \kappa^{-1}$, provided that $\kappa$ is nonzero a.e., with $X^\dagger$ denoting the Moore-Penrose pseudo-inverse of $X$.
    \end{itemize}
    \item \textbf{GLT 4.} \(\{A_n\}_n \sim_{\text{GLT}} \kappa\) if and only if there exist \(\{B_{n,j}\}_n \sim_{\text{GLT}} \kappa_j\) such that \(\{\{B_{n,j}\}_n\}_j \xrightarrow{\text{a.c.s.wrt j}} \{A_n\}_n\) and \(\kappa_j \rightarrow \kappa\) in measure.
    \item \textbf{GLT 5.} If $\{A_{{n}}\}_n \sim_{\text{GLT}} \kappa$ and $A_{{n}} = X_{{n}} + Y_{{n}}$, where
\begin{itemize}
    \item every $X_{{n}}$ is Hermitian,
    \item $ ||X_{{n}}||, \, ||Y_{{{n}}}|| \leq C$ for some constant $C$ independent of $n$,
    \item ${n}^{-1} \|Y_{{n}}\|_1 \to 0$,
\end{itemize}
then $\{A_{{n}}\}_n \sim_\lambda \kappa$.
   \item \textbf{GLT 6.}  If $\{A_{{n}}\}_n \sim_{\text{GLT}} \kappa$ and each $A_{{n}}$ is Hermitian, then $\{f(A_{{n}})\}_n \sim_{\text{GLT}} f(\kappa)$ for every continuous function $f : \mathbb{C} \to \mathbb{C}$.
\end{itemize}

Notice that a matrix-sequence $\{A_{{n}}\}_n$ as in Axiom \textbf{GLT 5.} is defined as quasi-Hermitian. The related theory developed in \cite{GoSe,BaSe} represents a useful way for overcoming the strict requirements of the Hermitian character of the matrices in the standard a.c.s. topology.
%second part of Theorem \ref{main theory}.

In the following sections, namely Section \ref{sec: CW-model} and Section \ref{sec: our-pb}, we prove that two remarkable matrix-sequences stemming from mean-field quantum spin systems are GLT matrix-sequences. In particular, a properly scaled matrix-sequence associated with the Curie-Weiss model is a zero-distributed Hermitian (real symmetric) matrix-sequence i.e. GLT matrix-sequence with $0$ GLT symbol, while the one associated with the restricted model is a GLT matrix-sequence with an interesting non-trivial GLT symbol.
For the sake of notational clarity, for the matrix-sequence in Section \ref{sec: CW-model}, the size is $d_N=2^N$.

\section{Curie-Weiss model}\label{sec: CW-model}
We consider the Hamiltonian for the quantum Curie-Weiss model for ferromagnetism, which takes the form
\begin{align}
    H_{\Lambda_N}^{\text{CW}}=-\frac{\Gamma}{2|\Lambda_N|}\sum_{x,y\in \Lambda_N}\sigma_3(x)\sigma_3(y)-B\sum_{x\in\Lambda_N}\sigma_1(x), \label{eq:curieweiss}
\end{align}
where $\Lambda_N$ is an arbitrary finite subset of $\mathbb{Z}^\ell$, $\Gamma>0$ scales the spin-spin coupling, and $B$ is an external magnetic field. This model describes a chain of $N$ immobile spin-$1/2$ particles with ferromagnetic coupling in a transverse magnetic field. The Hamiltonian acts on the Hilbert space $\mathcal{H}_{\Lambda_N}=\otimes_{x\in\Lambda_N}H_x$, where $H_x=\mathbb{C}^2$.
The operator $\sigma_i(x), \dots (i=1,2,3)$ acts as the Pauli matrix $\sigma_i$ on $H_x$ and acts as the unit matrix $\mathbbm{1}_2$ elsewhere. The single Pauli matrices are explcitly given by
\begin{align}\label{basic matrices}
    \sigma_1=\begin{pmatrix}
    0 & 1 \\ 1 & 0
\end{pmatrix},\ \ \ \sigma_2=\begin{pmatrix}
    0 & -i \\ i & 0
\end{pmatrix},\ \ \ \sigma_3=\begin{pmatrix}
    1 & 0 \\ 0 & -1
\end{pmatrix}.
\end{align}
\noindent
In contrast to locally interacting quantum spin models, the spatial dimension of this model does not influence the behaviour. This follows from the fact that for the averages
\begin{align}\label{spinss}
    S_i^{\Lambda_N}=\frac{1}{|\Lambda_N|}\sum_{x\in\Lambda_N}\sigma_i(x), \ \ \ (i=1,2,3),
\end{align}
we can write the Hamiltonian \eqref{eq:curieweiss} (see e.g. \cite{Ven_2024}) as
\begin{align}\label{hamcw}
    H_{\Lambda_N}^{\text{CW}}=-|\Lambda_N|\bigg{(}\frac{\Gamma}{2}(S_3^{\Lambda_N})^2+BS_1^{\Lambda_N})\bigg{)}=
    |\Lambda_N|\bigg{(}h_0^{\text{CW}}({\bf S}^{\Lambda_N})\bigg{)},
\end{align}
for ${\bf S}^{\Lambda_N}=(S_1^{\Lambda_N},S_2^{\Lambda_N},S_3^{\Lambda_N})$, and the choice
\begin{align}\label{classical CW}
B^3\ni (x,y,z)\mapsto h_0^{\text{CW}}(x,y,z)=-\bigg{(}\frac{\Gamma}{2}z^2+Bx\bigg{)},
\end{align}
with $B^3=\{(x,y,z)\ | \ x^2+y^2+z^2\leq 1\}$ the unit three-dimensional sphere in $\mathbb{R}^3$.
%We have used the Landau $\mathcal{O}$-notation, i.e., the error $\mathcal{O}(1/|\Lambda_N|)$ -which arises from the fact that $\sigma_i^2=\mathbbm{1}_2$- is a multiple of the identity and bounded by $C/N$ with a constant $C$ which only depends on the coefficients defining the polynomial $h_0^{\text{CW}}$.

Using spherical coordinates ${\bf e}(\Omega)=(\sin{\vartheta}\cos{\varphi},\sin{\vartheta}\sin{\varphi},\cos{\vartheta})$,  with $\vartheta\in [0,\pi]$ and $\varphi\in [0,2\pi)$; for radius $u\in [0,1]$, we may now rewrite
 \begin{align}\label{classicalCWmodel}
 h_0^{CW}(u\cdot {\bf e}(\Omega))=-\bigg{(}\frac{\Gamma}{2}(u\cos{\vartheta})^2+Bu\sin{\vartheta}\cos{\varphi}\bigg{)},
\end{align}
as being a polynomial function on $[0,1]\times\mathbb{S}^2$.

As a result, we may as well consider the quantum Curie–Weiss
Hamiltonian \eqref{eq:curieweiss} in $\ell=1$, so that we may simply write $|\Lambda_N|=N$ and $H_{N}^{\text{CW}}:= H_{\Lambda_N}^{\text{CW}}$.

\subsection{The CW model as GLT-sequence}
Here we show that the real symmetric matrix-sequence $\{\bar{H}_{N}^{\text{CW}}\}_N$, $d_N=2^N$, where $\bar{H}_{N}^{\text{CW}}:= H_{N}^{\text{CW}}/N$, is a basic GLT matrix-sequence with GLT symbol $0$, i.e., a zero-distributed matrix-sequence (Axiom \textbf{GLT 2.}, third item).

\begin{theorem}\label{main result}
The normalized  CW Hamiltonian $\bar{H}_{N}^{\text{CW}}:= H_{N}^{\text{CW}}/N$ defines $\{\bar{H}_{N}^{\text{CW}}\}_N$ as a zero-distributed GLT sequence.
\end{theorem}
\begin{proof}
The Curie-Weiss Hamiltonian is self-adjoint and in fact all the considered matrices are real symmetric. Thus in view of Definition \ref{99} (with  $d_N=2^N$, $D=[0,1]\times\mathbb{S}^2$), it suffices to prove that  for all real-valued $F\in C_c(\mathbb{R})$
\begin{align}\label{limit1}
    \lim_{N\to\infty}\frac{1}{2^N}Tr_{2^N}[F(\bar{H}_N^{\text{CW}})]=F(0),
\end{align}
where $Tr_{2^N}$ denotes the trace on $\mathcal{H}_N=\bigotimes_{i=1}^N\mathbb{C}^2$.
The convergence in \eqref{limit1} for all possible test functions $F$ implies that the sequence $(\bar{H}_{\Lambda_N}^{\text{CW}})_N$ is a zero-distributed GLT sequence.

Since the spectrum of the normalized Curie-Weiss Hamiltonian $\bar{H}_N^{CW}$ is uniformly bounded in $N$ and contained in a connected compact subset $C \subset \mathbb{R}$ (see also \cite{Ven_2020,Ven_2022} for further details), it suffices to consider functions $F$ supported on $C$. By the Stone-Weierstrass theorem, any continuous function $F$ on $C$ can be uniformly approximated by polynomials  restricted to $C$: in fact this argument is standard and was used e.g. in \cite{GoSe} in the context of the zero distribution of orthogonal polynomials, when the Jacobi operator is perturbed by a non self-adjoint compact operator. Moreover, since the (normalized) trace operation is linear and continuous with respect to uniform convergence of continuous functions on the spectrum, it suffices to prove \eqref{limit1} for all polynomials $P$ restricted to $C$. Thus, we will establish \eqref{limit1} for all such polynomials $P$, which by approximation will imply the result for any general $F \in C_c(\mathbb{R})$ supported on $C$.
%Moreover, we only focus on the principal part $h_0^{CW}({\bf S}^N)$, since, if \eqref{limit1} holds true for $h_0^{CW}({\bf S}^N)$, it also holds also true for $\bar{H}_N^{\text{CW}}=h_0^{CW}({\bf S}^N)+\mathcal{O}(1/N)$, by continuity.

To do so, we stress that the Hamiltonian is homogeneously decomposable \cite[Section 2]{Cegla_Lewis_Raggio_1988},
implying that the model is block diagonal with the following tracial decomposition
\begin{align*}
   Tr_{2^N}[\bar{H}_N^{\text{CW}}]=\frac{1}{2^N}\sum_{J\in\mathbb{J}_N}C(J,N)Tr_{2J+1}[\bar{H}_N^{\text{CW}}(J)],
\end{align*}
where, by \cite{Mihailov_1977}, the quantities
$$C(J,N)=\frac{2J+1}{N+1}\binom{N+1}{\frac{N}{2}+J+1},$$
are the multiplicities of the $(2J+1)$-dimensional irreducible unitary representations arising in the decomposition of the $N$-fold tensor product representation of $SU(2)$ onto $\mathbb{C}^2$ with itself. Here, $\mathbb{J}_N=\{0,1,...,N/2\}$ if $N/2$ is an integer, and equals $\{1/2,3/2,...,N/2
\}$ if $N/2$ is a half-integer.
The $(2J+1)\times (2J+1)$-dimensional matrix $\bar{H}_N^{\text{CW}}(J)$ is  defined as \begin{align}\label{restricted CW}
\bar{H}_N^{\text{CW}}(J):=
h_0^{CW}({\bf S}^N)|_{2J+1}.%+\mathcal{O}(1/N),
\end{align}
%where $\mathcal{O}(1/N)$ does not depend on $J$. As mentiond above, in the sequel we omit the $\mathcal{O}(1/N)$-term and we only focus on the principal part $h_0^{CW}({\bf S}^N)|_{2J+1}$.

Define the function
\begin{align}\label{newfunction}
    P_N(J):=P(\bar{H}_N^{\text{CW}}(J))=P(h_0^{CW}({\bf S}^N)|_{2J+1})=P(h_0^{CW}({\bf S}^N))|_{2J+1},
\end{align}
so that
\begin{align}\label{id0}
\frac{1}{2^N}Tr_{2^N}[P(h_0^{CW}({\bf S}^N))]=\frac{1}{2^N}\sum_{J\in\mathbb{J}_N}C(J,N)Tr_{2J+1}(P_N(J)).
\end{align}

This allows us to introduce a sequence of probability measures $\nu_N$, defined for Borel measurable sets $E\subset [0,1]$, by
\begin{align}\label{idprob}
    \nu_N(E)=\frac{1}{2^N}\sum_{\{J\in\mathbb{J}_N\ |\ 2J/N\in E\}}C(J,N)(2J+1).
\end{align}
Note that indeed $2^N=\sum_{\{J\ |\ 2J/N\in [0,1]\}}C(J,N)(2J+1)$, so that $\nu_N([0,1])=1$.
We recall the well-known {\em resolution of the identity} \cite{Ven_2020}, that is,
\begin{align}\label{resolution identity}
    \mathbbm{1}_{2J+1}=\frac{2J+1}{4\pi}\int_{\mathbb{S}^2}d\mu_{\mathbb{S}^2}(\Omega)Pr(J,\Omega),
\end{align}
where $d\mu_{\mathbb{S}^2}$ is the uniform measure on the 2-sphere, and $Pr(J,\Omega)$ is the one-dimensional projection onto the linear span of the $(2J+1)$-dimensional spin-coherent state vector labeled by the point $\Omega\in\mathbb{S}^2$.
Let $$p_N\bigg{(}\frac{2J}{N},\Omega\bigg{)}:=Tr_{2J+1}[P_N(J)Pr(J,\Omega)].$$

By inserting the identity \eqref{resolution identity}, we can write
\begin{align}\label{id1}
Tr_{2J+1}(P_N(J))=\frac{2J+1}{4\pi}\int_{\mathbb{S}^2}d\mu_{\mathbb{S}^2}(\Omega)Tr_{2J+1}[P_N(J)Pr(J,\Omega)]= \frac{2J+1}{4\pi}\int_{\mathbb{S}^2}d\mu_{\mathbb{S}^2}(\Omega)p_N\bigg{(}\frac{2J}{N},\Omega\bigg{)}.
\end{align}

% Let $r_N^P$ be a sequence of continuous functions on $[0,1]\times\mathbb{S}^2$ such that for all $N\in\mathbb{N}_+, J\in\mathbb{J}$ and $\Omega\in\mathbb{S}^2$, it holds $r_N^P\bigg{(}\frac{2J}{N},\Omega\bigg{)}=p_N\bigg{(}\frac{2J}{N},\Omega\bigg{)}$.
If we set $K_N(du,d\Omega):=\nu_N(du)\times \frac{1}{4\pi}\mu_{\mathbb{S}^2}(d\Omega)$, equations \eqref{id0}, \eqref{idprob} and \eqref{id1} imply that
$$\frac{1}{2^N}Tr_{2^N}[P(h_0^{CW}({\bf S}^N))]=\int_{[0,1]\times \mathbb{S}^2}K_N(du,d\Omega)p_N(u,\Omega),$$
where the integral over $[0,1]$ should be interpreted as the discrete integral with  respect to $\nu_N$. It is clear that $K_N$ is a probability meausure for each $N$.
Using the previous preparatory statements, we are now in a position to prove the following facts.
\begin{itemize}
    \item[(i)]
    For all $N\in\mathbb{N}_+$, it holds
\begin{align}\label{onemli1}
\sup_{J\in\mathbb{J}_N}\sup_{\Omega\in\mathbb{S}^2}\bigg{|}p_N\bigg{(}\frac{2J}{N},\Omega\bigg{)} - P\bigg{(}h_0^{\text{CW}}\bigg{(}\frac{2J}{N}{\bf e}(\Omega)\bigg{)}\bigg{)}\bigg{|}\leq \frac{C}{N},
\end{align}
    the constant $C$ being independent of $J$, $\Omega$ and $N$.
    %The sequence $(r_N^F(u,\Omega))_N$ converges uniformly as $N\to\infty$ towards the function $F(h_0^{\text{CW}}(u,\Omega))$;
    \item[(ii)] The sequence of probability measures $(K_N)_N$ converges setwise to the probability measure $\delta_0\times \frac{1}{4\pi}\mu_{\mathbb{S}^2}$, with $\delta_0$ the Dirac measure concentrated at $0\in [0,1]$.
\end{itemize}
% The conditions (i) and (ii) are sufficient to conclude that:
% $$\lim_{N\to\infty}\frac{1}{2^N}Tr_{2^N}[P(\bar{H}_N^{\text{CW}})] =\frac{1}{4\pi}\int_{\mathbb{S}^2}\mu_{\mathbb{S}^2}(d\Omega) P(h_0^{CW}(0\cdot \Omega))=P(0),$$
% since $h_0^{CW}(0\cdot\Omega)=0$ for all $\Omega\in\mathbb{S}^2$. This is precisely what we needed to show, cf. \eqref{limit1}.
% \\\\
\noindent
For (i), we first rely on the following fact obtained in \cite{Manai_Warzel}. For any  non-commuting self-adjoint polynomial $H(J):=P_0({\bf S}^N)|_{2J+1}$ on $\mathbb{C}^{2J+1}$, it holds
\begin{align}\label{onemli2}
    \sup_{J\in \mathbb{J}_N}\sup_{\Omega\in\mathbb{S}^2}\bigg{|}Tr_{2J+1}[Pr(J,\Omega)H(J)]-P_0\bigg{(}\frac{2J}{N}{\bf e}(\Omega)\bigg{)}\bigg{|}\leq \frac{C}{N},
    %&\sup_{J\in \mathbb{J}_N}\bigg{|}\bar{H}(J)-Q_{J}\bigg{(}P\bigg{(}\frac{2J}{N},\cdot\bigg{)}\bigg{)}\bigg{|}\leq \frac{C}{N},
\end{align}
where the constant $C>0$ only depends on the chosen $P_0$, not on $J$, $\Omega$ and $N$.

To prove \eqref{onemli1}, we stress that for any given polynomial $P$, the operator $P\circ h_0^{CW}({\bf S}^N)$ is again a polynomial in three non coummuting self-adjoint spin operators.

As a matter of fact, the statement \eqref{onemli1} follows from \eqref{onemli2} for the choice $P_0=P\circ h_0^{CW}$.

In order to see that (ii) holds, we rewrite
\begin{align}\label{nu}
    \nu_N(E)&=\frac{1}{2^N}\sum_{\{J\ |\ 2J/N\in E\}}\frac{(2J+1)^2}{N+1}\binom{N+1}{N/2+J+1}.
\end{align}
Define $X_N\sim Bin(N+1,1/2)$ and write $J=X_N-\frac{N}{2}-1$. Use the binomial probability mass function (pmf) $BinProb$,
$$BinProb(X_N=k)=\binom{N+1}{k}\frac{1}{2^{N+1}},$$ so that \eqref{nu} becomes
\begin{align*}
    \nu_N(E)&=2\sum_{\{X_N\ |\ \frac{2(X_N-\frac{N}{2}-1)}{N}\in E\}}\frac{(2(X_N-\frac{N}{2}-1)+1)^2}{N+1}BinProb(X_N).
\end{align*}
Consider $\epsilon>0$ arbitrary and let $E\subset [0,1]$ be a measurable set such that $E\cap [0,\epsilon]=\emptyset$.
We deduce that
$$\frac{2J}{N}\in E\ \ \ \implies \ \ \  X_N \in \bigg{(} \frac{N(\epsilon+1)+2}{2}, N+1\bigg{]}.$$
As a result,
\begin{align}\label{concentration}
&BinProb\bigg{(}\frac{2(X_N-\frac{N}{2}-1)}{N}\in E\bigg{)}\leq BinProb\bigg{(}X_N\geq \frac{N(\epsilon+1)+2}{2}\bigg{)}\nonumber\\&\leq \exp{\bigg{(}\frac{-(\frac{N\epsilon+1}{N+1})^2\frac{N+1}{2}}{3}\bigg{)}}\nonumber\\&=
\exp{\bigg{(}-\frac{(N\epsilon+1)^2}{6(N+1)}\bigg{)}},
\end{align}
where, in the second inequality, we have exploited Chernoff's inequality for the binomial distribution \cite[Theorem 4.4]{Mitz}. To derive this, consider $\mu=\frac{N+1}{2}$ and $p=1/2$, so that on account of the standard Chernoff inequality with $0<\delta< 1$, we find
$$BinProb(X_N\geq (1+\delta)\mu)\leq e^{-\frac{\delta^2}{3}\mu}.$$
If $X_N\sim Bin(N+1,1/2)$, then $\mu=\frac{N+1}{2}$, so that for the choice $\delta=(a-\mu)/\mu$ with $a=(N(\epsilon+1)+2)/2$ one indeed obtains \eqref{concentration}, as certainly $0<\delta< 1$.
The inequality \eqref{concentration} implies that the pmf $BinProb$ concentrates around zero, as $N\to\infty$, as long as $E$ is bounded away from zero.

It follows that for all $E$ with $E\cap [0,\epsilon]=\emptyset$,
\begin{align}\label{measures}
\nu_N(E)\leq 2(N+1)^2e^{-(N\epsilon+1)^2/6(N+1)}.
\end{align}
As a result, $(\nu_N)_N$ converges setwise to the Dirac mass concentrated at zero. Since the measure $\mu_{\mathbb{S}^2}$ does not depend on $N$, the same statement holds true for $K_N=\nu_N\times  \frac{1}{4\pi}\mu_{\mathbb{S}^2}$, i.e. $(K_N)_N$ converges setwise to $\delta_0\times  \frac{1}{4\pi}\mu_{\mathbb{S}^2}$. This shows the validity of $(ii)$.

We are finally in a position to prove \eqref{limit1}. To this avail, we notice that $P\circ h_0^{CW}$ is uniformly continuous. Hence, given $\varepsilon>0$ there is $\delta>0$ such that for all $u,u'\in [0,1]$ and ${\bf e}(\Omega), {\bf e}(\Omega')\in \mathbb{S}^2$, for which $||u{\bf e}(\Omega)-u'{\bf e}(\Omega')||< \delta$, it holds
\begin{align}\label{uniformly cts}
    |P(h_0^{CW}(u\cdot {\bf e}(\Omega))-P(h_0^{CW}(u'\cdot {\bf e}(\Omega'))|<\varepsilon.
\end{align}
We now estimate
\begin{align*}
    &\bigg{|}\frac{1}{4\pi}\int_{[0,1]\times \mathbb{S}^2}K_N(du,d\Omega)p_N(u,\Omega)-P(0)\bigg{|}\leq\\&\frac{1}{4\pi}\int_{[0,\delta)\times \mathbb{S}^2}K_N(du,d\Omega)|p_N(u,\Omega)-P(0)|+\frac{1}{4\pi}\int_{[\delta, 1]\times \mathbb{S}^2}K_N(du,d\Omega)|p_N(u,\Omega)-P(0)|\leq\\& \underbrace{\sup_{\substack{J\in\mathbb{J}_N \\ 0\leq \frac{2J}{N}< \delta}}\sup_{\Omega\in\mathbb{S}^2}\bigg{|}p_N\bigg{(}\frac{2J}{N},\Omega\bigg{)}-P(0)\bigg{|}}_{\text{(I)}} + \underbrace{\sup_{\substack{J\in\mathbb{J}_N \\ \delta \leq \frac{2J}{N}\leq 1}}\sup_{\Omega\in\mathbb{S}^2}\bigg{|}p_N\bigg{(}\frac{2J}{N},\Omega\bigg{)}-P(0)\bigg{|}\nu_N([\delta,1])}_{\text{(II)}},
\end{align*}
To estimate term (II), we notice that following bound holds, i.e.
$$ \sup_{\substack{J\in\mathbb{J}_N \\ \delta\leq \frac{2J}{N}\leq 1}}\sup_{\Omega\in\mathbb{S}^2}\bigg{|}p_N\bigg{(}\frac{2J}{N},\Omega\bigg{)}-P(0)\bigg{|}\leq 2\|P\|_\infty,$$
since  for all $J$, $N$ and $\Omega$, one has
$$p_N\bigg{(}\frac{2J}{N},\Omega\bigg{)}\leq \|P\|_\infty.$$
On account of \eqref{measures}, it holds $\nu_N([\delta,1])< 2(N+1)e^{-(N\delta+1)^2/6(N+1)}$. It follows that there exists $N$ sufficiently large, such that $(II)$ can be made smaller than $\varepsilon$.

For (I), we first exploit the triangle inequality: for all $N\in\mathbb{N}_+$, $J\in\mathbb{J}_N$ and $\Omega\in\mathbb{S}^2$, it holds
\begin{align}\label{triangle}
    \bigg{|}p_N\bigg{(}\frac{2J}{N},\Omega\bigg{)}-P(0)\bigg{|}\leq \bigg{|}p_N\bigg{(}\frac{2J}{N},\Omega\bigg{)}-P\bigg{(}h_0^{CW}\bigg{(}\frac{2J}{N}{\bf e}(\Omega)\bigg{)} \bigg{)} \bigg{|}+\bigg{|}P\bigg{(}h_0^{CW}\bigg{(}\frac{2J}{N}{\bf e}(\Omega)\bigg{)} \bigg{)}-P(0)\bigg{|}.
\end{align}
If we now apply the supremum over $J\in\mathbb{J}_N$, for which $0\leq\frac{2J}{N}< \delta$ and the supremum over $\Omega\in\mathbb{S}^2$, then, for $N$ large enough, the first summand in \eqref{triangle} is bounded by $\varepsilon$ on account of \eqref{onemli1}. The second summand is bounded by $\varepsilon$ due to uniform continuity  of $P\circ h_0^{CW}$, cf. \eqref{uniformly cts}, which is applied at  the zero vector, i.e., at the point $0\cdot {\bf e}(\Omega))={\bf 0}$, for which
it holds $h_0^{CW}(0\cdot {\bf e}(\Omega))=0$.

This shows the validity of \eqref{limit1}, thereby concluding the proof of the theorem.

\end{proof}

\begin{remark}
 We note that the previous result readily generalizes to any (normalized) mean-field quantum spin model expressed as a polynomial in the total spin operator. This follows from \cite[Proposition II.2]{Raggio_Werner_1989}, which states that the continuous functional calculus of a (normalized) mean-field model, representing a so-called quasi-symmetric sequence, remains quasi-symmetric. A detailed exploration of this generalization is left for future work.
 \hfill$\blacksquare$
\end{remark}

\section{Restricted Curie-Weiss model}\label{sec: our-pb}
Through direct inspection, the Curie-Weiss model $H_{N}^{\text{CW}}$ preserves the symmetric subspace $\text{Sym}^N(\mathbb{C}^2)$, which has dimension $N+1$. Consequently, the model can be restricted this subspace. With a slight abuse of notation, we denote this restricted $(N+1)\times (N+1)$ matrix by $H^{s}_N$, and, as before,  normalize it by the factor $1/N$, yielding $\Bar{H}^{s}_N$. This restricted matrix represents a single quantum spin system with spin quantum number $J=N/2$, a setting commonly analyzed in the classical limit  $J\to\infty$ \cite{Moretti_vandeVen_2020,Ven_2020}.  Specifically, the matrix aligns with the choice $J=N/2$ in \eqref{restricted CW}.
As demonstrated in the proof of Theorem \ref{main result}, the resulting family of matrices $\{\Bar{H}^{s}_N\}_N$  a Berezin-Toeplitz operator with symbol $h_0^{CW}\in C(\mathbb{S}^2)$, corresponding to the case $u=1$ in \eqref{classicalCWmodel}.
\begin{remark}\label{change variables}
    By making the following change of variables
$$\cos{\vartheta}\mapsto 2x-1;$$
$$\varphi\mapsto \theta:=\varphi-\pi,$$
it follows that the symbol (for $u=1$) reads
\begin{equation}\label{symbol restricted}
h_0^{CW}(x,\varphi)=-\frac{\Gamma}{2}(2x-1)^2 {-}2B\sqrt{(1-x)x}\cos{\theta}, \ \ \ x\in [0,1], \ \ \ \theta\in [-\pi,\pi].
\end{equation}
\hfill$\blacksquare$
\end{remark}
In the sequel, we step-by-step prove that the Berezin-Topelitz operator $\Bar{H}^{s}_N$ is indeed a GLT sequence. For achieving this, we recall that there exists a basis in which $\Bar{H}^{s}_N$ is represented by the following matrix
\begin{equation}\label{seq restricted}
    \Bar{H}^{s}_N = \underset{1 \leq k \leq N+1}{\text{diag}} \left( -\frac{\Gamma}{2} \left( \frac{2k}{N} - 1 \right)^2 \right)
+ \text{tridiag} \left( -B \sqrt{1-\frac{(k-1)}{N}} \sqrt{\frac{k}{N}} \quad 0 \quad -B \sqrt{1-\frac{k}{N}} \sqrt{\frac{k+1}{N}} \right),
\end{equation}
as proven in \cite{Ven_Groenenboom_Reuvers_Landsman}.
\begin{theorem}\label{main result-bis}
With reference to the setting in (\ref{symbol restricted}) and (\ref{seq restricted}), we have
\[
\{\Bar{H}^{s}_N \}_N \sim_{\rm{GLT},\sigma,\lambda} h_0^{CW}(x,\varphi(\theta)).
\]
\end{theorem}
\begin{proof}
By considering $d=r=1$ it is immediate to see that the matrix $\Bar{H}^{s}_N$ can be written as the sum of diagonal sampling matrices as in Section \ref{ssec: sampl} and two products of sampling matrices and very basic unilevel Toeplitz matrices as in Section \ref{ssec: toep}. In fact we have
\begin{equation*}
 \Bar{H}^{s}_N = -\frac{\Gamma}{2}D_{N+1}\Big(\Big(2x-1\Big)^2\Big)-B D_{N+1}(\sqrt{1-x}\sqrt{x})T_{N+1}(e^{\iota \theta})-B T_{N+1}(e^{-\iota \theta}) D_{N+1}(\sqrt{1-x}\sqrt{x}).
\end{equation*}
Now by Axiom \textbf{GLT 2.1} we have $\{T_{N+1}(e^{\pm \iota \theta})\}_N\sim_{\text{GLT}} e^{\pm \iota \theta}$ and by Axiom \textbf{GLT 2.2} we deduce
\[
\left\{D_{N+1}\Big(\Big(2x-1\Big)^2\Big)\right\}_N\sim_{\text{GLT}}  \left( 2x-1 \right)^2, \ \ \
\{D_{N+1}(\sqrt{1-x}\sqrt{x})\}_N\sim_{\text{GLT}}  \sqrt{(1-x)x},
\]
simply because both functions $\left( 2x-1 \right)^2, \sqrt{1-x}\sqrt{x}$ are continuous on $[0,1]$ and a fortiori Riemann integrable.

Hence by using the $*$-algebra structure of the GLT sequences and more precisely Axiom \textbf{GLT 3.2}, Axiom \textbf{GLT 3.3}, we infer $\, \{\Bar{H}^{s}_N\}_N \sim_{\text{GLT}} -\frac{\Gamma}{2}  \left( 2x-1 \right)^2 -2B\cos(\theta)\sqrt{(1-x)x}$, which is compatible with the symbol
$h_0^{CW}(x,\varphi(\theta))$ indicated in Remark \ref{change variables}.
Finally \textbf{GLT 1.} implies
\begin{itemize}
    \item $\, \{\Bar{H}^{s}_N\}_N \sim_{\sigma} -\frac{\Gamma}{2}  \left( 2x-1 \right)^2 -2B\cos(\theta)\sqrt{(1-x)x}$.
    \begin{itemize}
        \item Furthermore, since $\Bar{H}^{s}_N$ is real and symmetric for any $N$, again by \textbf{GLT 1.}, we deduce
    \end{itemize}
    \item $\, \{\Bar{H}^{s}_N\}_N \sim_{\lambda} -\frac{\Gamma}{2}  \left( 2x-1 \right)^2 -2B\cos(\theta)\sqrt{(1-x)x}$.
\end{itemize}

Notice that $\, \{\Bar{H}^{s}_N+U_N\}_N \sim_{\text{GLT}} -\frac{\Gamma}{2}  \left( 2x-1 \right)^2 -2B\cos(\theta)\sqrt{(1-x)x}$ for any
$\{U_{n}\}_N$ such that  $\{U_{n}\}_N \sim_{\text{GLT}} 0$  by Axiom \textbf{GLT 3.2} and  Axiom \textbf{GLT 2.3} so that
$\, \{\Bar{H}^{s}_N+U_N\}_N \sim_{\sigma} -\frac{\Gamma}{2}  \left( 2x-1 \right)^2 -2B\cos(\theta)\sqrt{(1-x)x}$. Finally, as it happens for compact non-Hermitian perturbations of Jacobi matrix-sequences \cite{GoSe}, $\, \{\Bar{H}^{s}_N+U_N\}_N \sim_{\lambda} -\frac{\Gamma}{2}  \left( 2x-1 \right)^2 -2B\cos(\theta)\sqrt{(1-x)x}$ either if all $U_N$ are Hermitian or if the assumption in Axiom \textbf{GLT 5.} is satisfied by the non-Hermitian matrix-sequence perturbation $\{U_N\}_N$.
\end{proof}
% $+ U_{N}$ with$\{U_{n}\} \sim_{\text{GLT}} 0;$

%\subsection{Relation with a discretized Schr\"{o}dinger operator}
\begin{remark}\label{Schr\"{o}dinger operator}
It is interesting to comment the relations between a standard second order centered finite difference (FD) discretization of the Schr\"{o}dinger operator and Theorem \ref{main result-bis}.
Let us define
$a(x):=2B\sqrt{(1-x)x}$ and $c(x):=-\frac{\Gamma}{2}(2x-1)^2-2B\sqrt{(1-x)x}$, i.e. we have set $\cos{\theta}=1$ in the definition of the symbol. For $N\in\mathbb{N}_+$, consider the Schr\"{o}dinger operator on $L^2(0,1)$:
$$H(N)u(x)=-\frac{1}{(N+1)^2}a(x)u''(x)+c(x)u(x)$$
with $u(0)=u(1)$. If we discretize the domain in uniform steps of width $h=1/(N+1)$, it follows that the resulting FD discretization matrix is a GLT sequence with symbol $a(x)(2-2\cos{\theta})+c(x)$. The factor $1/N$ simultaneously plays the role of the discretization step-size as well as the semi-classical parameter. It is clear that the ensuing GLT matrix-sequence is equivalent to the restricted Curie-Weiss model. Hence, the restricted Curie-Weiss model defines a Schr\"{o}dinger operator with potential $c$. In the parameter regime  $\Gamma=1$ and $B=1$, $c$  has the shape of a single well, cf. Section \ref{harmregime}, whilst for the choices  $\Gamma=1$ and $B\in (0,1)$, $c$  has the shape of a double wel, cf. Section \ref{schrregime}.
However, the fact that the parameter factor $1/N$ plays  simultaneously the role of the discretization step-size and of the semi-classical parameter is non-standard from the numerical analysis viewpoint and the whole potential of a related theoretical study has still to be explored further.
\end{remark}

\section{Numerical results}\label{sec: num}

In the present section, we give various visualizations of the spectral features of the matrices $\Bar{H}^{s}_N$, confirming the derivations in Theorem \ref{main result-bis}. We fix the value of $\Gamma$ and $B$ and, for these fixed values, we consider the matrix-size parameter $N$ equal to $40, 80, 160, 320$. We recall that Theorem \ref{main result-bis} is an asymptotic one, as all the GLT results, but the really impressive fact is that the spectrum of  $\Bar{H}^{s}_N$ adheres to the spectral symbol already for $N=40$.

\subsection{Asymptotic spectral behavior of $\Bar{H}^{s}_N$ with $\Gamma=B=1$}\label{harmregime}

As already mentioned, we consider the matrix
\begin{equation*}
 \Bar{H}^{s}_N = -\frac{\Gamma}{2}D_{N+1}\Big(\Big(2x-1\Big)^2\Big) {-}B D_{N+1}(\sqrt{1-x}\sqrt{x})T_{N+1}(e^{\iota \theta}){-}B T_{N+1}(e^{-\iota \theta}) D_{N+1}(\sqrt{1-x}\sqrt{x}),
\end{equation*}
whose associated matrix-sequence has proven to have a GLT nature in Section \ref{sec: our-pb} in Theorem \ref{main result-bis}.

We show the comparison between symbol $h_0^{CW}$ (for $u=1$), cf. \eqref{classicalCWmodel},  and the eigenvalues of  $\Bar{H}_N^s$. This comparison is carried out using the concept of {\em monotone rearrangement}, which enables the interpretation of the symbol as a single-variable function
\cite[Definition 2.1]{BarBianGar}.

We notice that the agreement is very good, even for moderate matrix-sizes, which is nontrivial given the asymptotic nature of the GLT distributional results.
A further remarkable fact is that the range of the spectral symbol $h_0^{CW}$ contains all the spectra, i.e., no outliers are observed. Again, this is a highly nontrivial matter and cannot be deduced directly from distributional results. In fact, the interplay between spectral localization and distributional properties is characteristic of linear positive operators (LPOs), such as Toeplitz operators and various classes of variable-coefficient coercive differential operators (see \cite[Corollary 6.2]{garoni2017} for the Toeplitz case and \cite{cma-rev} for a broader discussion).

\begin{figure}[H]
\centering
\includegraphics[scale=0.24]{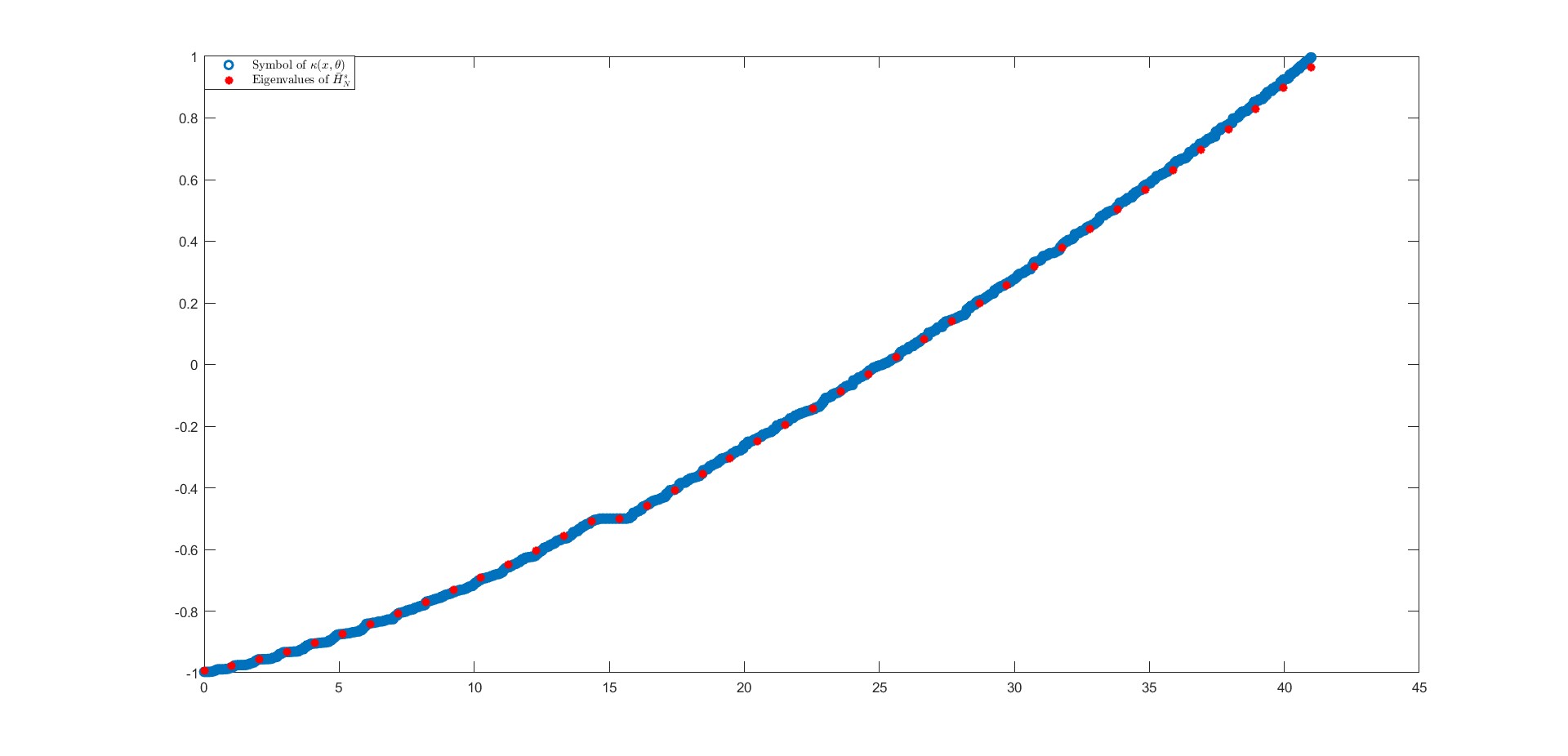}
\caption*{$N=40$}
\end{figure}
\begin{figure}[H]
\centering
\includegraphics[scale=0.24]{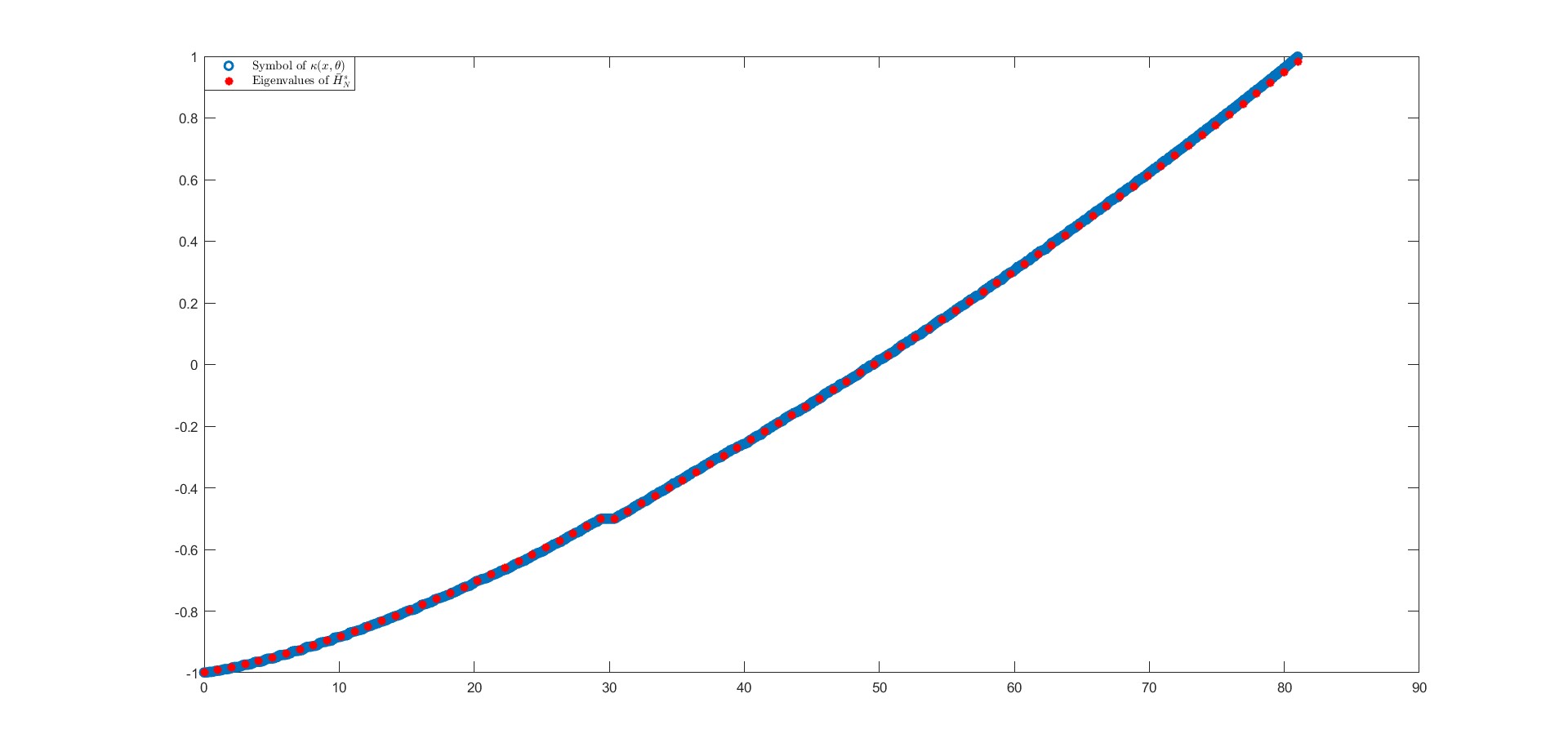}
\caption*{$N=80$}
\end{figure}
\begin{figure}[H]
\centering
\includegraphics[scale=0.24]{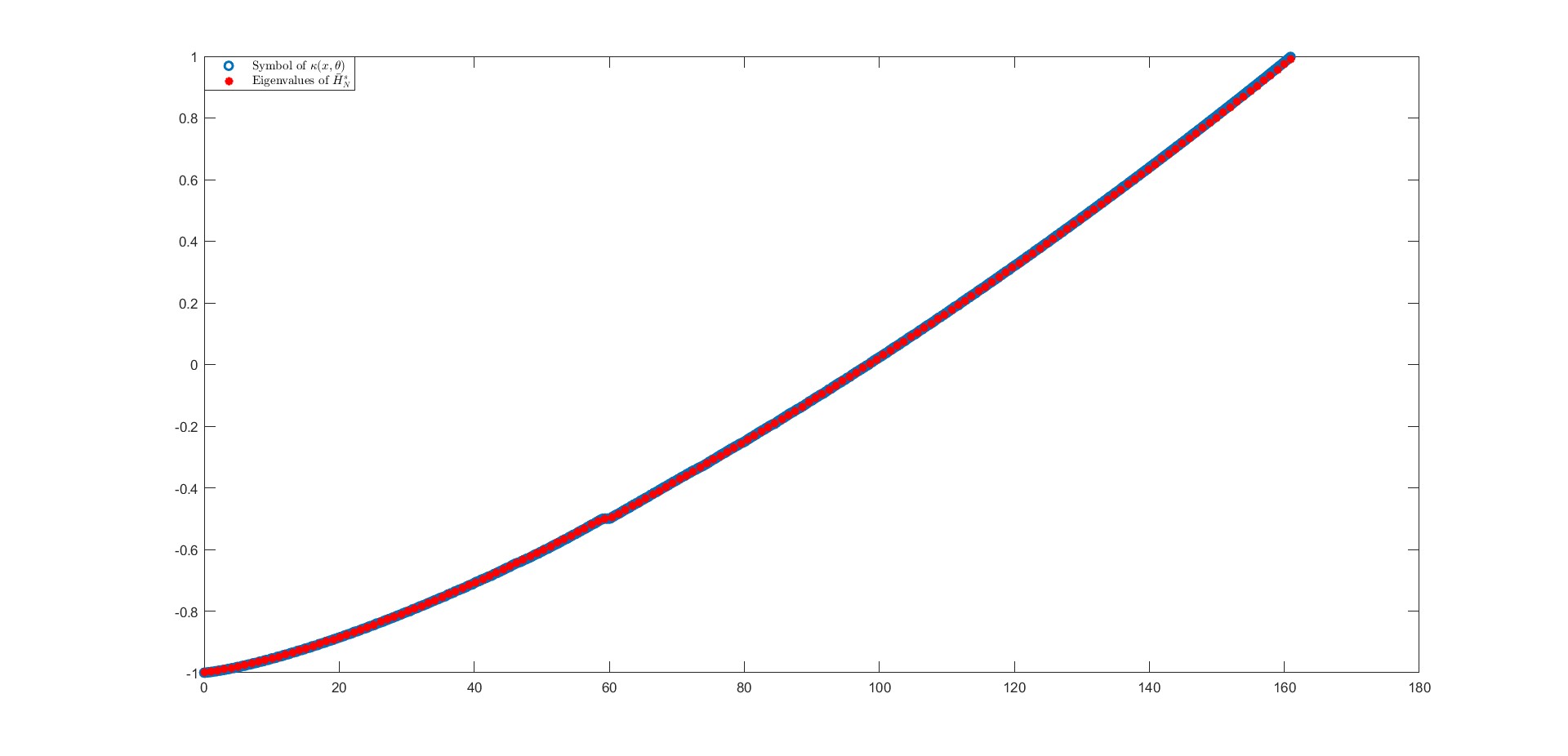}
\caption*{$N=160$}
\end{figure}
\begin{figure}[H]
\centering
\includegraphics[scale=0.24]{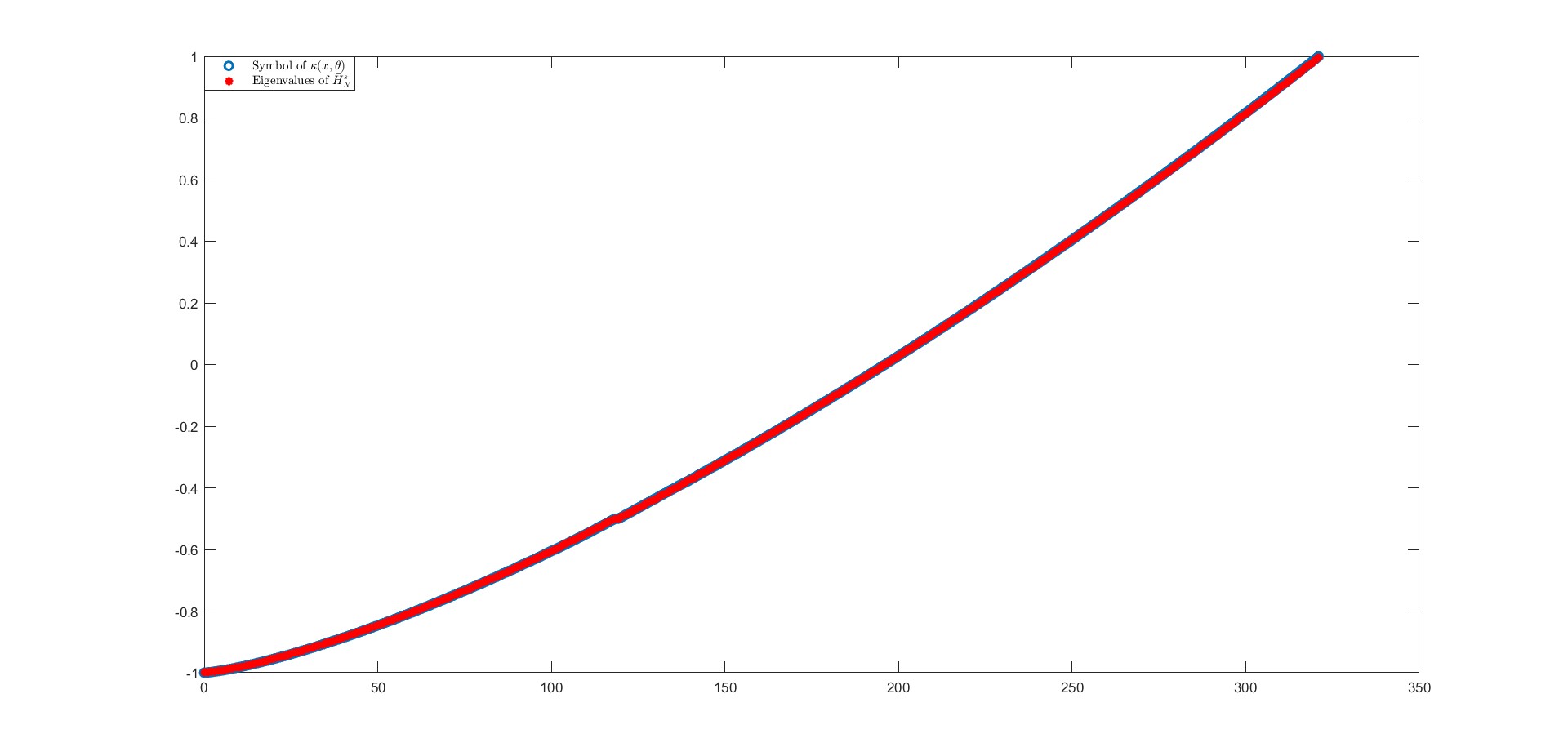}
\caption*{$N=320$}
\end{figure}

\subsection{Asymptotic spectral behavior of $\Bar{H}^{s}_N$ with $\Gamma=2B=1$}\label{schrregime}

{For $\Gamma=1$ and $B=\frac{1}{2}$ we show again the comparison between symbol and the eigenvalues.} The comments already provided in the case $\Gamma=B=1$ can be repeated verbatim also in this setting.

\begin{figure}[H]
\centering
\includegraphics[scale=0.24]{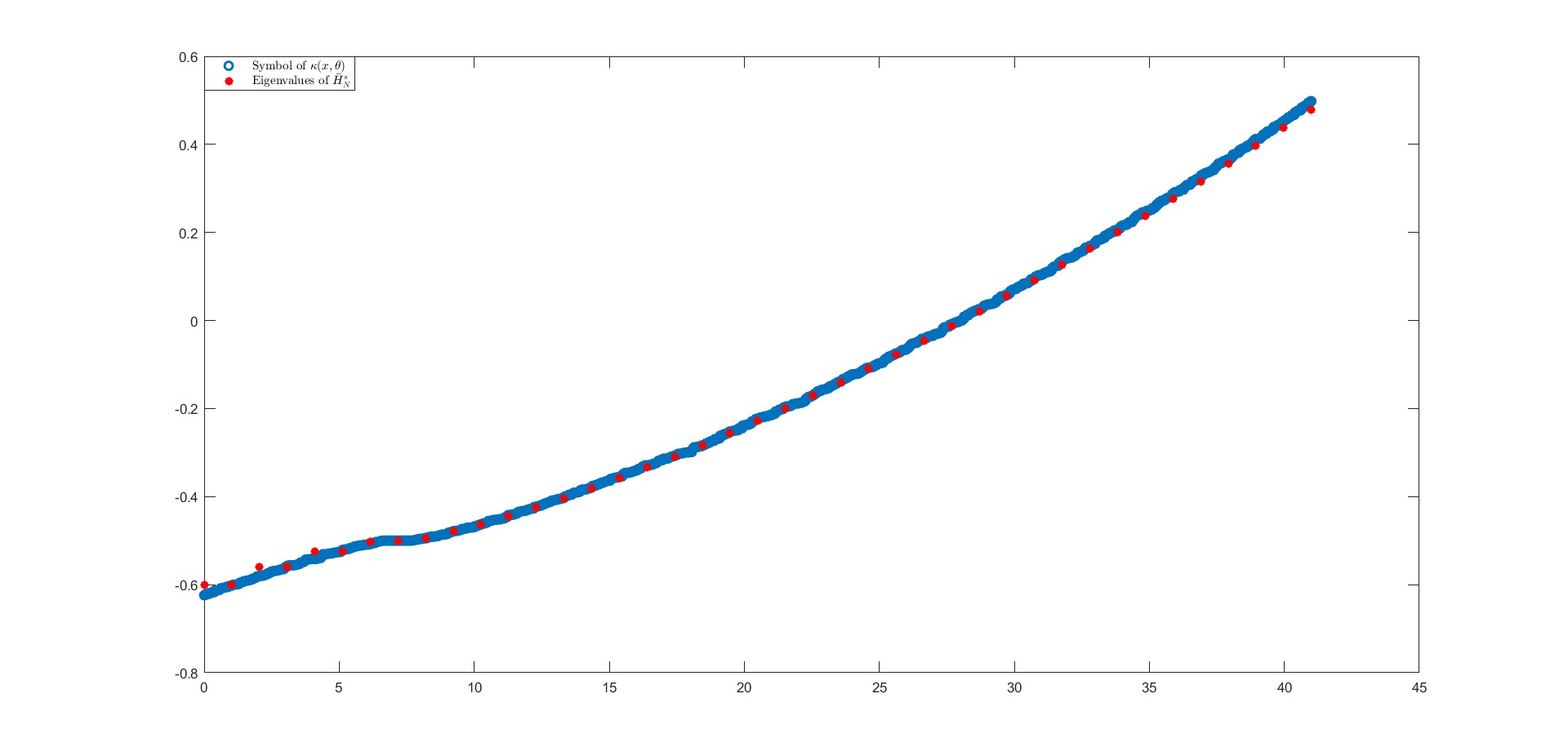}
\caption*{$N=40$}
\end{figure}
\begin{figure}[H]
\centering
\includegraphics[scale=0.24]{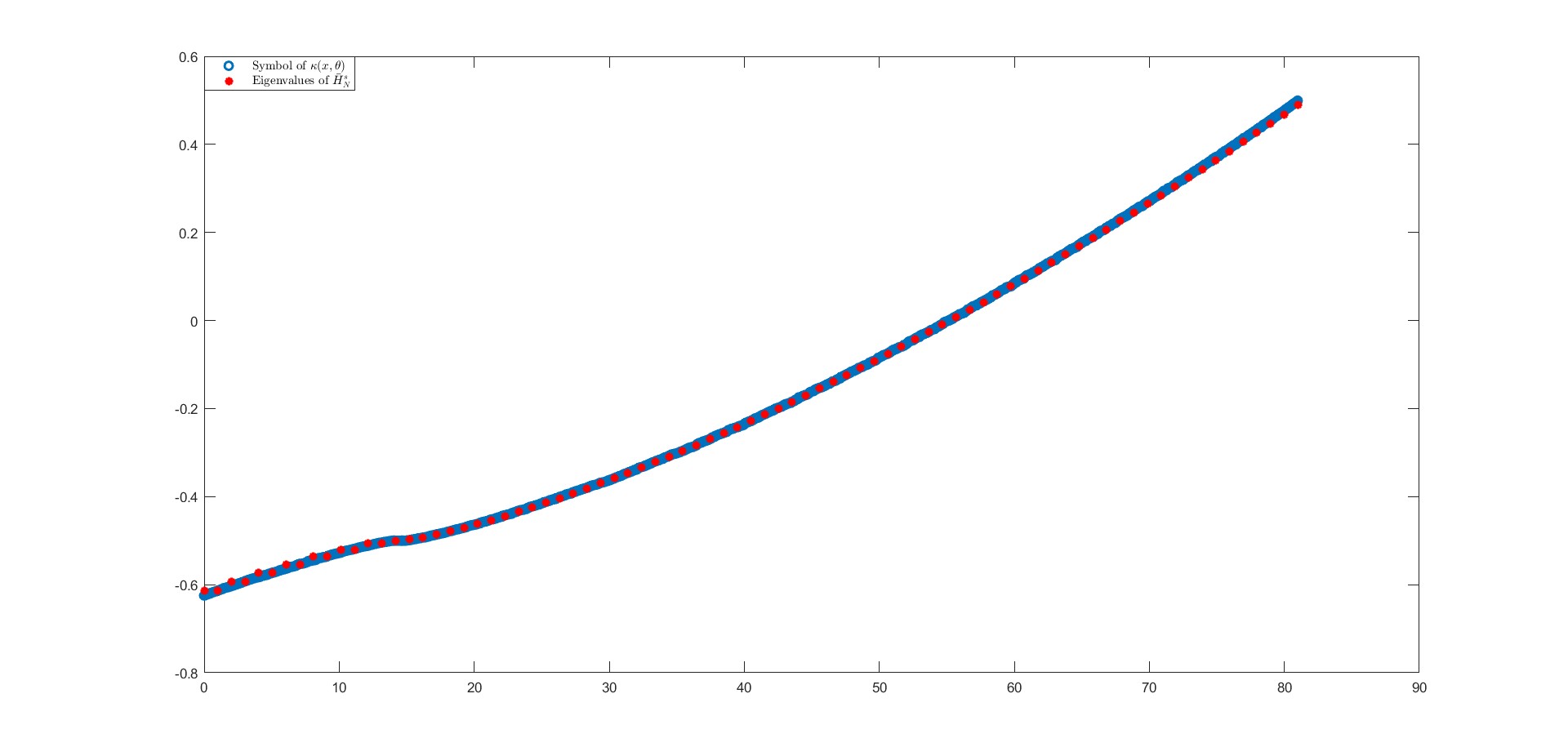}
\caption*{$N=80$}
\end{figure}
\begin{figure}[H]
\centering
\includegraphics[scale=0.24]{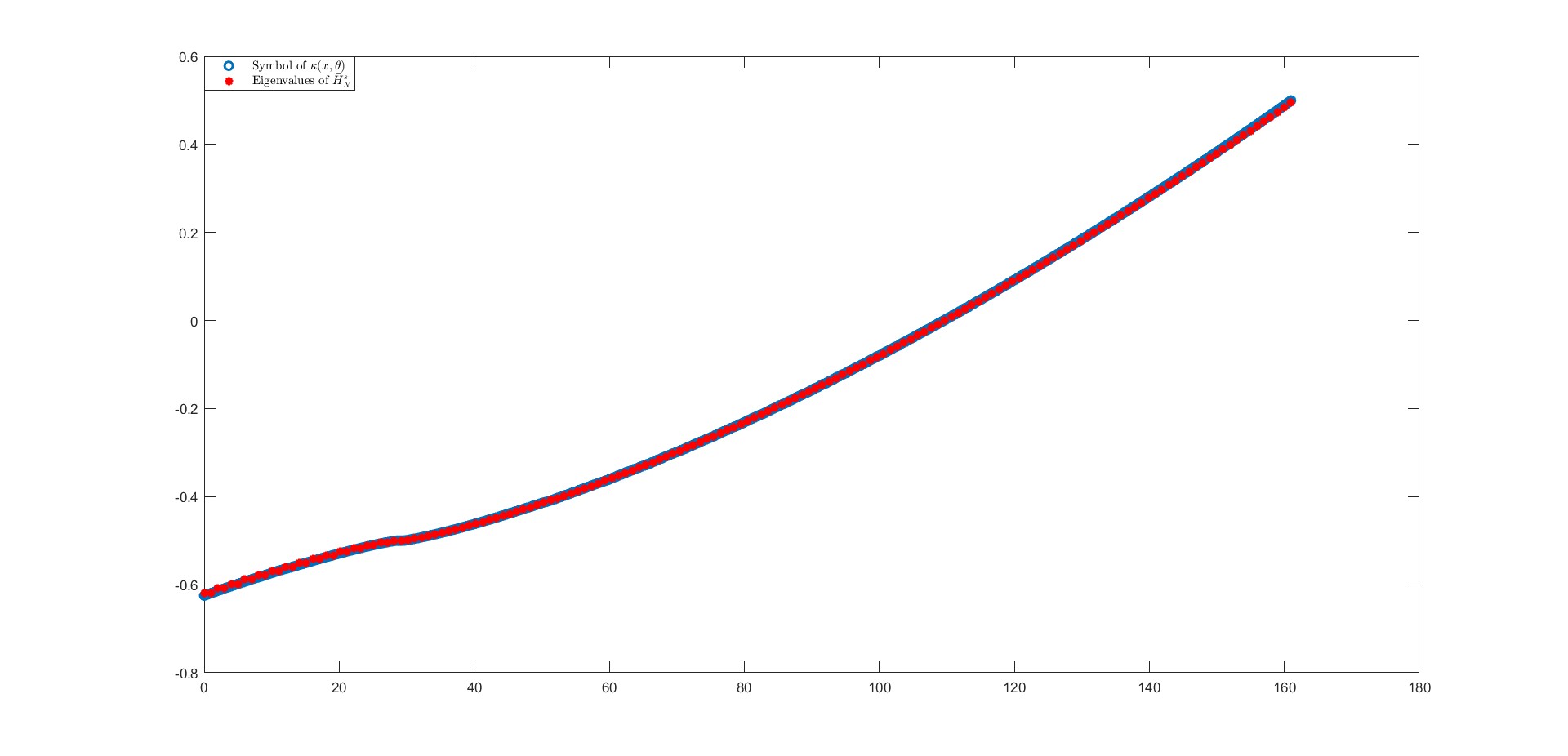}
\caption*{$N=160$}
\end{figure}
\begin{figure}[H]
\centering
\includegraphics[scale=0.24]{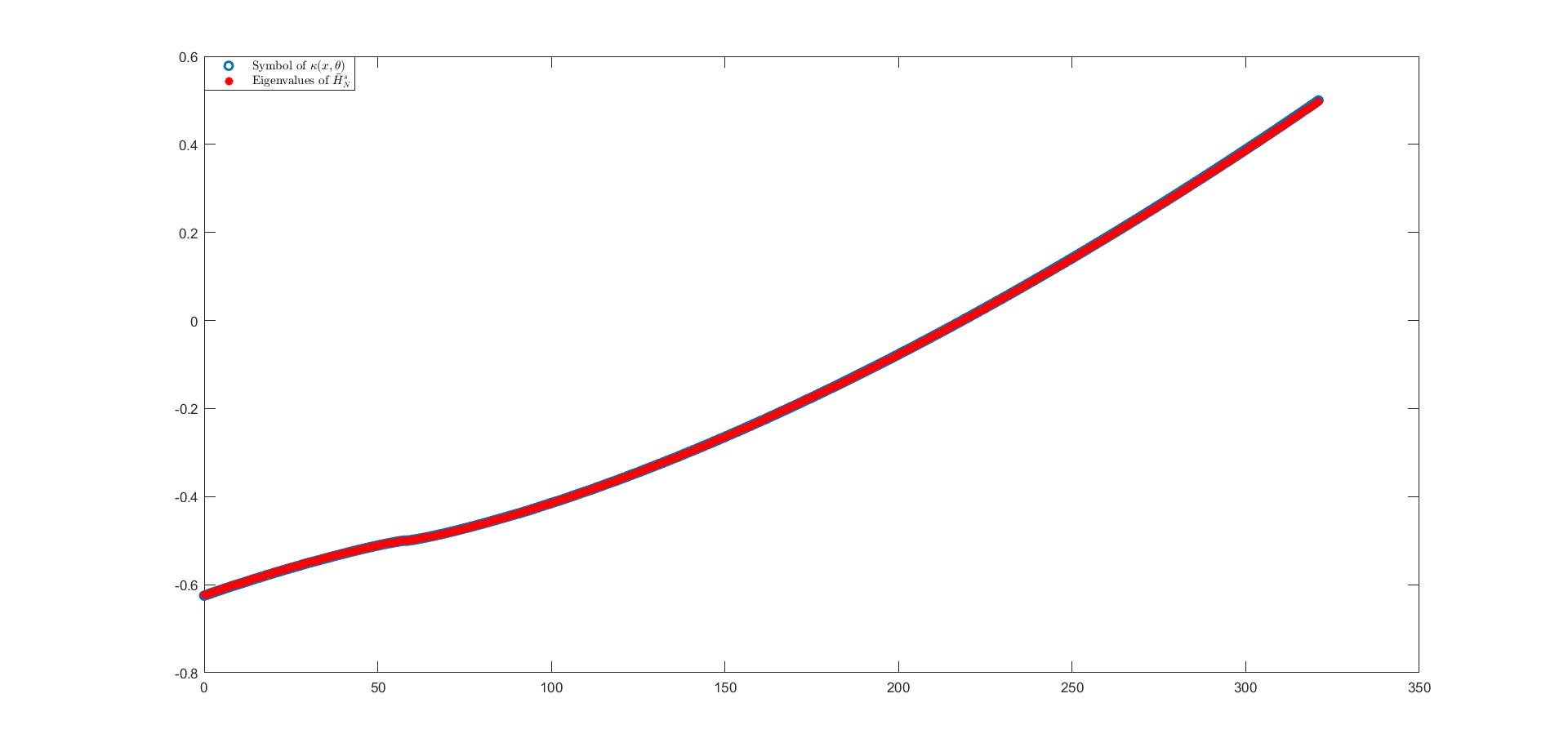}
\caption*{$N=320$}
\end{figure}

The figures provide clear evidence that the lowest eigenvalues become nearly doubly degenerate,  a characteristic feature of  a Schr\"{o}dinger operator with a double-well potential. For a more detailed discussion on the approximation of this Schr\"{o}dinger operator by the Curie-Weiss model, we refer the reader to \cite{Ven_Groenenboom_Reuvers_Landsman}.

\subsection{Extremal spectral behavior of $\Bar{H}^{s}_N$}

The figures reported in the previous section inform of localization results, beyond the proven distributional results. Here we show that in the two considered setting of parameters ($\Gamma=B=1$ and  $\Gamma=2B=1$), not only the spectrum of $\Bar{H}^{s}_N$ is contained in the interior of the range of the GLT symbol, but the behavior of the extreme eigenvalues is very regular. In fact, both the minimal and the minimal eigenvalues converge monotonically to the minimum and to maximum of the symbol, respectively: again this phenomenon is typical of matrix-valued LPOs \cite{cma-rev}.

For the minimum/maximum eigenvalue analysis, for $\Gamma=B=1$ and we proceed as follows.

\begin{itemize}
    \item The GLT symbol is $\kappa(x,\theta)=-\frac{\Gamma}{2}(2x-1)^{2}- 2B\cos(\theta)\sqrt{(1-x)x}$.
    \item We consider $m=\underset{(x,\theta) \in [0,1]\times [0,\pi]}{\text{min}} \kappa(x,\theta)$,  $M=\underset{(x,\theta) \in [0,1]\times [0,\pi]}{\text{max}} \kappa(x,\theta)$.
\end{itemize}
%\subsubsection{}
\begin{itemize}
    \item Take $N_{j} +1=40*2^{j}, \quad\quad j=0,1,2,3, $
    \item Compute $\tau_{j}= \lambda_{\min}(\bar{H}_{N_{j}}^{s})-m, \quad \quad j=0,1,2,3, $
    \item Compute $\alpha_{j}=\log\left(\frac{\tau_{j}}{\tau_{j+1}}\right), \quad \quad j=0,1,2. $
\end{itemize}
\vspace{9pt}
\begin{table}[H]
\centering
\caption{$\Gamma=1$, $B=1$ and $m=-1$}
\begin{tabular}{c|c|c|c}
   & $\lambda_{\min}(\bar{H}_{N_{j}}^{s})$ & $\tau_{j}$ & $\alpha_{j}$ \\
  \hline
 $N_{j} +1=40$ & -0.9936 & 0.0064 & 0.4082 \\
  \hline
  $N_{j} +1=80$& -0.9975 & 0.0025 & 0.3979 \\
  \hline
  $N_{j} +1=160$ & -0.9990 & 0.001 & 0.3979 \\
  \hline
  $N_{j} +1=320$ & -0.9996 & 0.0004 &   \\
\end{tabular}
\end{table}

%\subsubsection{}
\begin{itemize}
    \item Take $N_{j} +1=40*2^{j}, \quad\quad j=0,1,2,3, $
    \item Compute $\hat{\tau}_{j}= M-\lambda_{\max}(\bar{H}_{N_{j}}^{s}), \quad \quad j=0,1,2,3, $
    \item Compute $\beta{j}=\log\left(\frac{\hat{\tau}_{j}}{\hat{\tau}_{j+1}}\right), \quad \quad j=0,1,2. $
\end{itemize}
\vspace{9pt}
\begin{table}[H]
\centering
\caption{$\Gamma=1$, $B=1$ and $M=1$}
\begin{tabular}{c|c|c|c}
   & $\lambda_{\max}(\bar{H}_{N_{j}}^{s})$ & $\hat{\tau}_{j}$ & $\beta_{j}$ \\
  \hline
 $N_{j} +1=40$ & 0.9654 & 0.0346 & 0.2960 \\
  \hline
  $N_{j} +1=80$& 0.9825 & 0.0175 & 0.2985 \\
  \hline
  $N_{j} +1=160$ & 0.9912 & 0.0088 & 0.3010 \\
  \hline
  $N_{j} +1=320$ & 0.9956 & 0.0044 &   \\
\end{tabular}
\end{table}

The interpretation of the above tables is very informative. In fact, if we consider $\psi$ the nondecreasing rearrangement of the GLT symbol $\kappa(x,\theta)$ defined on the standard interval $[0,1]$, then the numerical tests suggest that
\[
\lambda_{\min}(\bar{H}_{N_{j}}^{s})-m \sim {1\over N^{0.4}}, \ \ \ \ M-\lambda_{\max}(\bar{H}_{N_{j}}^{s}) \sim {1\over N^{0.3}}.
\]
It remains to investigate analytically if the corresponding analytic behavior holds that is
\[
\psi(t)-m \sim {t^{0.4}}, \ \ \ \ M-\psi(t)  \sim {(1-t)^{0.3}}.
\]

We now continue, along the previous reasoning, {with the computation of the minimum/maximum eigenvalues analysis, we proceed as follows in the case where $\Gamma=2B=1$.}
\begin{itemize}
    \item {The GLT symbol is $\kappa(x,\theta)=-\frac{\Gamma}{2}(2x-1)^{2}-2B\cos(\theta)\sqrt{(1-x)x}$}
    \item {We consider $m=\underset{(x,\theta) \in [0,1]\times [0,\pi]}{\text{min}} \kappa(x,\theta)$  $;$  $M=\underset{(x,\theta) \in [0,1]\times [0,\pi]}{\text{max}} \kappa(x,\theta)$}
\end{itemize}

%\subsubsection{{}}
\begin{itemize}
    \item {Take $N_{j} +1=40*2^{j}, \quad\quad j=0,1,2,3, $}
    \item {Compute $\tau_{j}= \lambda_{\min}(\Bar{H}_{N_{j}}^{s})-m, \quad \quad j=0,1,2,3, $}
    \item {Compute $\alpha_{j}=\log\left(\frac{\tau_{j}}{\tau_{j+1}}\right), \quad \quad j=0,1,2. $}
\end{itemize}

\vspace{9pt}
\begin{table}[H]
\centering
\caption{{$\Gamma=1$, $B=\frac{1}{2}$ and $m=-0.6241$}}
\begin{tabular}{c|c|c|c}
   & {$\lambda_{\min}(\Bar{H}_{N_{j}}^{s})$} & {$\tau_{j}$} & {$\alpha_{j}$} \\
  \hline
 {$N_{j} +1=40$} &{-0.6007} & {0.0234} & {0.3521} \\
  \hline
  {$N_{j} +1=80$} & {-0.6137} & {0.0104} & {0.3542} \\
  \hline
  {$N_{j} +1=160$} & {-0.6195} & {0.0046} & {0.4074} \\
  \hline
  {$N_{j} +1=320$} & {-0.6223} & {0.0018} &  \\
\end{tabular}
\end{table}

%\subsubsection{\textcolor{red}{}}
\begin{itemize}
    \item {Take $N_{j} +1=40*2^{j}, \quad\quad j=0,1,2,3, $}
    \item {Compute $\hat{\tau}_{j}= M-\lambda_{\max}(\Bar{H}_{N_{j}}^{s}), \quad \quad j=0,1,2,3, $}
    \item {Compute $\beta{j}=\log\left(\frac{\hat{\tau}_{j}}{\hat{\tau}_{j+1}}\right), \quad \quad j=0,1,2. $}
\end{itemize}

\vspace{9pt}
\begin{table}[H]
\centering
\caption{{$\Gamma=1$, $B=\frac{1}{2}$ and $M=0.4982$}}
\begin{tabular}{c|c|c|c}
   & {$\lambda_{\max}(\Bar{H}_{N_{j}}^{s})$} & {$\hat{\tau}_{j}$} & {$\beta_{j}$} \\
  \hline
 {$N_{j} +1=40$} & {0.4789} & {0.0193} & {0.3361} \\
  \hline
  {$N_{j} +1=80$} & {0.4893} & {0.0089} & {0.3930} \\
  \hline
  {$N_{j} +1=160$} & {0.4946} & {0.0036} & {0.4010} \\
  \hline
  {$N_{j} +1=320$} & {0.4973} & {0.0009} & {} \\
\end{tabular}
\end{table}

The interpretation of the numerical results is of interest. Again, as observed for the previous case, both the minimal and the minimal eigenvalues converge monotonically to the minimum and to maximum of the symbol, respectively, in line with results which are usually observed for matrix-valued LPOs.

\section{Conclusions}\label{sec: end}

In our study we have used the theory of GLT $*$-algebras for dealing with structured matrix-sequences, stemming from the modelling of mean-field quantum spin systems. More precisely, we have expressed the related matrix-sequences in the GLT formalism. Two cases have been considered in detail. In both cases, we have found the spectral distributions in the context of the GLT theory and the theoretical results have been confirmed via visualizations and numerical tests.
Many open problems remain. Among them we can cite the study of the extremal eigenvalues and the expansion of the considered matrix-sequences via GLT momentary GLT symbols \cite{momentary-1,momentary-2}, in order to have a finer description of the spectrum with respect to Theorem \ref{main result-bis}. For instance, in the nonreduced case, the GLT symbol equal to zero in Theorem \ref{main result} could hyde a finer structure in the sense that it could be imagined that there exists a numerical sequence $\alpha_N$ converging to zero such that the considered matrix-sequence divided by $\alpha_N$ is still a GLT sequence with nonzero GLT symbol. In that setting, it is possible that the emerging GLT structure could be multilevel (given the geometry on the sphere as in (\ref{classical CW})), block (given the presence of the basic blocks in (\ref{basic matrices})) \cite{barbarino2020uni, barbarino2020multi, garoni2017, garoni2018} or even of reduced type, since the terms in the general model have a varying sizes $C(J,N)$; see \cite{GLT-1}[pp. 398-399], \cite{GLT-2}[Section 3.1.4]  for the original idea and \cite{reduced} for an exhaustive treatment of reduced GLT matrix-sequences: it is clear that such a more precise result would provide a substantial improvement with respect to Theorem \ref{main result}.

Finally, it would be desirable to explore GLT theory in connection with locally interacting quantum spin models, such as the quantum Ising or quantum Heisenberg models, which represent realistic interacting systems commonly found in condensed matter physics. Progress in this direction would significantly contribute to our understanding of the fundamental structure of matter.

\section*{Acknowledgments}
%The authors would like to thank the Editor and the anonymous referees for their careful reading and for their helpful and pertinent suggestions.
The research of Stefano Serra-Capizzano is supported by the PRIN-PNRR project “MATH-ematical tools for predictive maintenance and PROtection of CULTtural heritage (MATHPROCULT)” (code P20228HZWR, CUP J53D23003780006), by INdAM-GNCS Project “Tecniche numeriche per problemi di grandi dimensioni” CUP_E53C24001950001, and by the European High-Performance Computing Joint Undertaking (JU) under Grant Agreement 955701. The JU receives support from the European Union’s Horizon 2020 research and innovation programme and Belgium, France, Germany, Switzerland. Furthermore Stefano Serra-Capizzano is grateful for the support of the Laboratory of Theory, Economics and Systems – Department of Computer Science at Athens University of Economics and Business. Finally Muhammad Faisal Khan and Stefano Serra-Capizzano are partly supported by Italian National Agency INdAM-GNCS.


\begin{thebibliography}{9}


\bibitem{gacs}
A. Adriani, A.J.A. Schiavoni-Piazza, and S. Serra-Capizzano, \emph{Block structures, g.a.c.s. approximation, and distributions}, Bol. Soc. Mat. Mex. 31, no. 2, paper 41 (2025)




\bibitem{momentary-2}
N.~Barakitis, V.~Loi, and S.~Serra-Capizzano, \emph{A note on eigenvalues and singular values of variable Toeplitz matrices and matrix-sequences, with application to variable two-step BDF approximations to parabolic equations}, {Springer book series ``Trends in Mathematics", in press (2025)}

% \bibitem{barbarino-equivalence} G. Barbarino, \emph{Equivalence between GLT sequences and measurable functions}, Linear Algebra Appl., 529 (2017), pp. 397–412.


    \bibitem{reduced} G. Barbarino, \emph{A systematic approach to reduced GLT}, BIT  62, no. 3, pp. 681--743 (2022)

     \bibitem{BarBianGar} G. Barbarino, D. Bianchi, and C. Garoni, \emph{Constructive approach to the monotone rearrangement of functions}, Expo. Math. 40,  no. 1, pp. 155--175 (2022)


     \bibitem{barbarino2020uni} G. Barbarino, C. Garoni, and S. Serra-Capizzano, \emph{Block generalized locally Toeplitz sequences: theory and applications in the unidimensional case}, Electr. Trans. Numer. Anal. 53, pp. 28--112 (2020)

    \bibitem{barbarino2020multi} G. Barbarino, C. Garoni, and S. Serra-Capizzano, \emph{Block generalized locally Toeplitz sequences: theory and applications in the multidimensional case}, Electr. Trans. Numer. Anal. 53, pp. 113--216 (2020)

\bibitem{BaSe}
G. Barbarino and S. Serra-Capizzano, \emph{Non-Hermitian perturbations of Hermitian matrix-sequences and applications to the spectral analysis of the numerical approximation of partial differential equations}, Numer. Linear Algebra Appl. 27, no. 3, paper e2286, 31 pp. (2020)


\bibitem{Appl-1-GLT}
P. Benedusi, P. Ferrari, M.E. Rognes, and S. Serra-Capizzano, {\em  Modeling excitable cells with the EMI equations: spectral analysis and iterative solution strategy}, J. Sci. Comput. 98, no. 3, paper 58, 23 pp. (2024)


\bibitem{momentary-1}
M. Bolten, S.-E. Ekström, I. Furci,  and S. Serra-Capizzano, \emph{A note on the spectral analysis of matrix sequences via GLT momentary symbols: from all-at-once solution of parabolic problems to distributed fractional order matrices}, Electr. Trans. Numer. Anal. 58, pp. 136--163 (2023)


\bibitem{Cegla_Lewis_Raggio_1988} W. Ceg\l a, J.T. Lewis, and G.A. Raggio. {\em The free energy of quantum spin systems and large
deviations.} Comm. Math. Phys. 118, pp. 337--354 (1988)


%\bibitem{Charles} B.L. Charles, {\em Berezin-Toeplitz Operators, a Semi-Classical Approach}, Vol. 239, pp. 1--28, (2003)

\bibitem{Appl-2-GLT}
A. Dorostkar, M. Neytcheva, and S. Serra-Capizzano, {\em Spectral analysis of coupled PDEs and of their Schur complements via generalized locally Toeplitz sequences in 2D}, Comput. Methods Appl. Mech. Engrg. 309, pp. 74–-105 (2016)


\bibitem{garoni-topology} C. Garoni, \emph{Topological foundations of an asymptotic approximation theory for sequences of matrices with increasing size}, Linear Algebra Appl., 513, pp. 324–341 (2017)

\bibitem{tom} C. Garoni, H. Speleers, S.-E. Ekström, A. Reali, S. Serra-Capizzano, and T.J.R. Hughes,  \emph{Symbol-based analysis of finite element and isogeometric B-spline discretizations of eigenvalue problems: exposition and review}, Arch. Comput. Methods Eng. 26, no. 5, pp. 1639--1690 (2019)


\bibitem{garoni2017} C. Garoni and S. Serra-Capizzano, \emph{Generalized locally Toeplitz sequences: theory and applications. Vol. I}, Springer, Cham, (2017)
\bibitem{garoni2018} C. Garoni and S. Serra-Capizzano, \emph{Generalized locally Toeplitz sequences: theory and applications. Vol. II}, Springer, Cham, (2018)

\bibitem{GoSe}
L. Golinskii and S. Serra-Capizzano, \emph{The asymptotic properties of the spectrum of nonsymmetrically perturbed Jacobi matrix sequences}, J. Approx. Theory 144, no. 1, pp. 84--102 (2007)

%\bibitem{Lieb} E.H. Lieb, \emph{The classical limit of quantum spin systems}, Commun. math. Phys. 31, 327340 (1973)

\bibitem{Manai_Warzel} C. Manai and S. Warzel, {\em The spectral gap and low-energy spectrum in mean-field quantum spin systems}, Forum Math. Sigma 11, paper No. e112, 34 pp. (2023)




\bibitem{Mihailov_1977} V.V. Mihailov, {\em Addition or arbitrary number of identical angular momenta}, J. Phys. A: Math. Gen. 10, no 2, paper 147 (1977)



\bibitem{Mitz}
M. Mitzenmacher and M. Jerrum, \emph{Probability and computing: randomized algorithms and probabilistic analysis}, Cambridge University Press, Cambridge, (2005)
%M. Talagrand, {\em Concentration Inequalities.} Springer-Verlag (2014)


\bibitem{Moretti_vandeVen_2020} V. Moretti and C.J.F. van de Ven, {\em Bulk-boundary asymptotic equivalence of two strict deformation quantizations}, Lett. Math. Phys. 110, no. 11, pp. 2941–-2963 (2020)



\bibitem{Raggio_Werner_1989}
G.A. Raggio and R. Werner, \emph{The statistical mechanics of general mean-field systems}, Helv. Phys. Acta 62, no. 8, pp. 98--1003 (1989)


\bibitem{Appl-3-GLT}
E. Salinelli, S. Serra-Capizzano, and D. Sesana, {\em Eigenvalue-eigenvector structure of Schoenmakers–Coffey matrices via Toeplitz technology and applications}, Linear Algebra Appl. 491, pp. 138-–160 (2016)

\bibitem{cma-rev}
S. Serra Capizzano, \emph{Some theorems on linear positive operators and functionals and their applications}, Comput. Math. Appl. 39, no. 7-8, pp. 139--167 (2000)

\bibitem{acs}
S. Serra Capizzano, \emph{ Distribution results on the algebra generated by Toeplitz sequences: a finite-dimensional approach}, Linear Algebra Appl. 328, no 1-3, pp. 121--130 (2001)

\bibitem{GLT-1}
S. Serra Capizzano, \emph{Generalized locally Toeplitz sequences: spectral analysis and applications to discretized partial differential equations}, Special issue on structured matrices: analysis, algorithms and applications (Cortona, 2000), Linear Algebra Appl. 366, pp. 371--402 (2003)

\bibitem{GLT-2}
S. Serra-Capizzano, \emph{The GLT class as a generalized Fourier analysis and applications}, Linear Algebra Appl. 419, no. 1, pp. 180--233 (2006)


\bibitem{Tilli-LT}
P. Tilli, {\em  Locally Toeplitz sequences: spectral properties and applications}, Linear Algebra Appl. 278, no. 1-3, pp. 91-–120 (1998)


\bibitem{Ven_2020} C.J.F. van de Ven, {\em The classical limit of mean-field quantum spin systems}, J. Math. Phys. 61, paper 121901 (2020)


\bibitem{Ven_2022} C.J.F. van de Ven, {\em The classical limit and spontaneous symmetry breaking in algebraic quantum theory}, Expo. Math. 40, no. 3, pp. 543–-571 (2022)

\bibitem{Ven_2024} C.J.F. van de Ven, {\em Gibbs states and their classical limit},  Rev. Math. Phys. 36, no. 5,  paper 2450009, 38 pp. (2024)

\bibitem{Ven_Groenenboom_Reuvers_Landsman} C.J.F. van de Ven, G, Groenenboom, R. Reuvers, and K. Landsman, {\em Quantum spin systems versus Schroedinger operators: A case study in spontaneous symmetry breaking}, SciPost Phys. 8, no. 2, paper 22, 36 pp. (2020)
\end{thebibliography}
\end{document}